\documentclass{article}

\usepackage[T1]{fontenc}
\usepackage{braket}
\usepackage{amssymb}
\usepackage{amsmath}
\usepackage{accents}
\usepackage{amsthm}
\usepackage{amsfonts}
\usepackage{mathtools}
\usepackage{xspace}
\usepackage{bm}
\usepackage{nicefrac}
\usepackage{commath}
\usepackage{bbm}
\usepackage{boxedminipage}
\usepackage{xparse}
\usepackage{xcolor}
\usepackage{float}
\usepackage{multirow}
\usepackage{graphicx}
\usepackage{ifthen}
\usepackage{adjustbox}
\usepackage{algorithm}
\usepackage{algpseudocode}
\usepackage{colortbl} 
\usepackage{fancyhdr}   
\usepackage{hyperref}
\hypersetup{ colorlinks=true, linkcolor=blue, filecolor=magenta, urlcolor=red, citecolor=blue, breaklinks}
\usepackage{fullpage}
\usepackage{environ}
\usepackage{thm-restate}
\usepackage[thinc]{esdiff}

\usepackage[capitalize]{cleveref}

\usepackage{tikz}
\usetikzlibrary{shapes}
\usetikzlibrary{positioning}
\usetikzlibrary{fit}
\usetikzlibrary{calc}
\usetikzlibrary{backgrounds}
\usetikzlibrary{patterns}
\usetikzlibrary{matrix}
\usetikzlibrary{arrows}

\newtheorem{proposition}{Proposition}
\newtheorem{lemma}{Lemma}[section]
\newtheorem{theorem}{Theorem}
\newtheorem{corollary}{Corollary}[section]
\newtheorem{definition}{Definition}

\theoremstyle{remark}
\newtheorem{remark}{Remark}

\newcommand{\ie}{\text{i.e.}\xspace}

\newcommand{\eg}{\text{e.g.}\xspace}
\newcommand{\cf}{\text{cf.}\xspace}

\newcommand{\tuple}[1]{\ensuremath{\left(#1\right)}}
\newcommand{\tp}{\tuple}
\newcommand{\eps}{\ensuremath{\varepsilon}}

\newcommand{\floor}[1]{\ensuremath{\left\lfloor{#1}\right\rfloor}\xspace}

\def\*#1{\mathbf{#1}}
\def\+#1{\mathcal{#1}}

\newcommand{\superscript}[1]{\ensuremath{^{\mbox{\tiny{\textit{#1}}}}}\xspace}
\def \th {\superscript{th}}     

\NewEnviron{problem}[1]{%
	\begin{center}\fbox{\parbox{6in}{%
				{\centering\scshape #1\par}%
				\parskip=1ex
				\everypar{\hangindent=1em}%
				\BODY
}}\end{center}}

\newcommand{\PRs}[2]{\ensuremath{\Pr_{#1}\left[#2\right]}}

\newcommand{\Exs}[2]{\ensuremath{\underset{#1}{\mathbb{E}}\left[#2\right]}}

\newcommand{\poly}{\ensuremath{\mathrm{poly}}\xspace}

\newcommand{\tv}[1]{\ensuremath{d_{\mathrm{TV}}\tp{#1}}}
\newcommand{\mle}{\mathrm{MLE}}

\newcommand{\tx}{\widetilde{x}}

\newcommand{\idc}[1]{\mathbf{1}\Set{#1}}

\def\final{1}  

\def\iflong{\iffalse}
\ifnum\final=0  
\newcommand{\mnote}[1]{{\color{purple} (Minshen's note: #1)}}
\newcommand{\enote}[1]{{\color{orange} (Elena's note: #1)}}
\newcommand{\xnote}[1]{{\color{magenta} (Xin's note: #1)}}
\newcommand{\msnote}[1]{{\color{blue} (Madhu's note: #1)}}
\newcommand{\knote}[1]{{\color{cyan} (Kuan's note: #1)}}
\newcommand{\todo}[1]{{\color{red} (TODO: #1)}}
\else 
\newcommand{\mnote}[1]{}
\newcommand{\enote}[1]{}
\newcommand{\xnote}[1]{}
\newcommand{\msnote}[1]{}
\newcommand{\knote}[1]{}
\newcommand{\todo}[1]{}
\fi

\title{On $k$-Mer-Based and Maximum Likelihood Estimation Algorithms for Trace Reconstruction}
\author{Kuan Cheng \thanks{Peking University, Haidian, Beijing, China.  Email: \texttt{ckkcdh@pku.edu.cn}.}
\and Elena Grigorescu \thanks{Purdue University, West Lafayette, IN, USA. Supported in part by NSF CCF-1910411, and NSF CCF-2228814. Email: \texttt{elena-g@purdue.edu}. }
\and Xin Li \thanks{Johns Hopkins University, Baltimore, MD, USA. Supported in part by NSF CAREER Award CCF-1845349 and NSF Award CCF-2127575. Email: \texttt{lixints@cs.jhu.edu}.}
\and Madhu Sudan\thanks{School of Engineering and Applied Sciences, Harvard University, Cambridge, Massachusetts, USA. Supported in part by a Simons Investigator Award and NSF Award CCF 2152413. Email: \texttt{madhu@cs.harvard.edu}.}
\and Minshen Zhu \thanks{Most of the work was done as a PhD student at Purdue University. Supported in part by NSF CCF-1910411, and NSF CCF-2228814. Email: \texttt{minshen.zh@gmail.com}. }
}
\date{}

\begin{document}

\maketitle

\begin{abstract}
    The goal of the trace reconstruction problem is to recover a string $\*x\in\{0,1\}^n$ given many independent {\em traces} of $\*x$, where a trace is a subsequence obtained from deleting bits of $\*x$ independently with some given probability $p\in [0,1).$ A recent result of Chase (STOC 2021) shows how $\*x$ can be determined (in exponential time) from $\exp({O}(n^{1/5})\log^5 n)$ traces. This is the state-of-the-art result on the sample complexity of trace reconstruction. 

   In this paper we consider two kinds of algorithms for the trace reconstruction problem. 
    
    We first observe that the bound of Chase, which is based on statistics of arbitrary length-$k$ subsequences, can also be  obtained by considering the ``$k$-mer statistics'', \ie, statistics regarding occurrences of {\em contiguous} $k$-bit strings (a.k.a, {\em $k$-mers}) in the initial string $\*x$, for $k = 2n^{1/5}$.  Mazooji and Shomorony  (arXiv.2210.10917) show that such statistics (called $k$-mer density map) can be estimated within $\eps$ accuracy from  $\poly(n, 2^k, 1/\eps)$ traces. 
    We call an algorithm to be {\em $k$-mer-based} if it reconstructs $\*x$ given estimates of the  $k$-mer density map.
    Such algorithms essentially capture all the analyses in the worst-case and smoothed-complexity models of the trace reconstruction problem we know of so far.
        
    Our first, and technically more involved, result shows that any $k$-mer-based algorithm for trace reconstruction must use $\exp(\Omega(n^{1/5} \sqrt{\log n}))$ traces, under the assumption that the estimator requires $\poly(2^k, 1/\eps)$ traces, thus establishing the optimality of this number of traces. The analysis of this result also shows that the analysis technique used by Chase (STOC 2021) is essentially tight, and hence new techniques are needed in order to improve the worst-case upper bound. 
    
    This result is shown by considering an appropriate class of real polynomials, that have been previously studied in the context of trace estimation (De, O'Donnell, Servedio. Annals of Probability 2019; Nazarov, Peres. STOC  2017), and proving that two of these polynomials are very close to each other on an arc in the complex plane. Our proof of the proximity of such polynomials uses new technical ingredients that allow us to focus on just a few coefficients of these polynomials. 

    Our second, simple, result considers the performance of the Maximum Likelihood Estimator (MLE), which specifically picks the source string that has the maximum likelihood to generate the samples (traces). We show that 
    the MLE algorithm uses a nearly optimal number of traces, \ie, up to a factor of $n$ in the number of samples needed for an optimal algorithm, and show that this factor of $n$ loss may be necessary under general ``model estimation'' settings.

\end{abstract}

\newpage

\section{Introduction}
The trace reconstruction problem is an infamous question introduced by Batu, Kannan, Khanna and McGregor \cite{BatuKKM04} in the context of computational biology. It asks to design algorithms that recover a string $\*x\in \{0,1\}^n$ given access to {\em traces} $\tilde{\*x}$ of $\*x$, obtained by deleting each bit independently with some given probability $p\in [0,1).$ The best current upper and lower bounds are exponentially apart, namely $\exp(\widetilde O(n^{1/5}))$ traces are sufficient for reconstruction \cite{chase2021separating} (improving upon the  $\exp(O(n^{1/3}))$ of \cite{nazarov2017trace, de2019optimal}) and ${\widetilde{\Omega}}(n^{3/2})$ \cite{HoldenL18, chase2021new} are necessary.

The problem has been recently studied in several variants so far \cite{BatuKKM04, Kannan005, ViswanathanS08, HolensteinMPW08, McGregorPV14, PeresZ17, nazarov2017trace, de2019optimal, GabrysM17,  HoldenPP18, HoldenL18, HartungHP18, GabrysM19, cheraghchi2020coded, KrishnamurthyM019, brakensiek2020coded, chen2020polynomialtime, chase2021separating, chase2021approximate, narayanan2021circular, sima2021trace, grigorescu2022limitations, rubinstein2023average} and it continues to elicit interest due to its deceptively simple formulation, as well as its motivating applications to DNA computing \cite{Olgica17}.

In this paper, we focus on the worst-case formulation of the problem, which is equivalent from an information-theoretic point of view to the {\em distinguishing} variant. In this variant, the goal is to distinguish whether the received  traces come from string $\*x\in \{0,1\}^n$ or from  $\*y\in \{0,1\}^n$, for some known $\*x\ne \*y.$

\paragraph{Algorithms based on $k$-bit statistics }

A very natural kind of algorithms \cite{HolensteinMPW08, nazarov2017trace,de2019optimal} operates using the mean of the received traces at each location $i\in [n]$ (one may assume that traces of smaller length than $n$ are padded with $0$'s at the end). Indeed, let $\+D_{\*x}$ be the distribution of the traces induced by the deletion channel on input $\*x$.
A mean/$1$-bit-statistics -based algorithm first estimates from the received traces the mean vector $\*E(\*x)=\tp{E_0(\*x), \cdots, E_{n-1}(\*x)} \in [0,1]^n$, where the $j$-th coordinate is defined as 
\begin{align*}
	E_j(\*x) = \Exs{\tilde{\*x} \sim \+D_{\*x}}{\tx_j}.
\end{align*}
It then may perform further post-processing without further inspection of the traces. 

Solving the distinguishing problem then reduces by standard arguments to understanding the $\ell_1$-norm between the mean traces of $\*x$ and $\*y$, namely the number $T$ of traces satisfies 

\[ \Omega\tp{1/\norm{\*E(\*x)-\*E(\*y)}_{\ell_1}} =T=O\tp{1/\norm{\*E(\*x)-\*E(\*y)}_{\ell_1}^2}.\]

\cite{nazarov2017trace, de2019optimal}  related the $\ell_1$-norm above with the supremum of a certain real univariate polynomial over the complex plane. Using techniques from complex analysis they proved that mean-based algorithms using $\exp(O(n^{1/3}))$ traces and outputting the string $\*s\in\{\*x, \*y\}$ whose $\*E(\*s)$  is closer in $\ell_1$-distance to the estimate is a successful reconstruction algorithm. Furthermore, any mean-based algorithm needs $\exp(\Omega(n^{1/3}))$ traces to succeed with high probability \cite{nazarov2017trace, de2019optimal}.

A general class of algorithms may operate by using $k$-bit statistics \cite{chase2021separating}, for $k\geq 1$.  Specifically, for $w\in \{0,1\}^k$, the algorithm estimates from the given traces, for  tuples $0\le i_0 < i_1 < \dots < i_{k-1}\le n-1$, the quantity 
\begin{align*}
 \Exs{\tilde{\*x} \sim \+D_{\*x}}{\prod_{{0\leq j < k}} \idc{\tx_{i_j}=w_j}}.
\end{align*}
After the estimation step, whose accuracy can be argued via standard Chernoff bounds, the algorithm does not need the traces anymore and may perform further post-processing in order to output the correct string.
The result of Chase follows from showing that for $k=2n^{1/5}$ there is a string $w\in \{0,1\}^{k}$ for which the $\ell_1$-distance between the corresponding $k$-bit statistics between $\*x$ and $\*y$ is large.

\paragraph{Algorithms based on $k$-mer statistics} Another variant proposed by Mazooji and Shomorony \cite{mazooji2022substring} considers algorithms which operate using estimates of statistics regarding occurrences of {\em contiguous} $k$-bit strings (a.k.a, {\em $k$-mers}) in the {\em initial} string $\*x$. We denote by $\idc{\*x[j\colon j+k-1] = w}$ the indicator bit of whether $w\in\{0,1\}^k$ occurs as a subword in $\*x$ from position $j$.

In particular, \cite{mazooji2022substring}  made the following definition which is central to our paper.
\begin{definition}[\cite{mazooji2022substring}]
Given a string $\*x \in \set{0,1}^n$ and a $k$-mer $w \in \set{0,1}^k$, for $i=0,1,\dots,n-1$ denote
\begin{align*}
    K_{w,\*x}[i] \coloneqq \sum_{j=0}^{n-k} \binom{j}{i}p^{j-i}(1-p)^{i} \cdot \idc{\*x[j\colon j+k-1] = w}.
\end{align*}
The vector $K_{\*x}\coloneqq \tp{K_{w,\*x}[i]\colon w \in \set{0,1}^k, i \in [n]}$ is called the \emph{$k$-mer density map} of $\*x$.
\end{definition}

Note that the mean vector $\*E(\*x)$ is, up to a factor of $1-p$, equivalent to the $1$-mer density map.  Indeed, for $k=1$ and $w=1$ we have 
\begin{align*}
E_i(\*x) &= \Exs{\tilde{\*x} \sim \+D_{\*x}}{ \tx_i} = \sum_{j=0}^{n-1}\Pr\left[\tx_i\textup{ comes from }x_j\right] \cdot x_j \\
&= \sum_{j=0}^{n-1}\binom{j}{i}p^{j-i}(1-p)^{i+1} \cdot x_j = (1-p)\cdot\sum_{j=0}^{n-1}\binom{j}{i}p^{j-i}(1-p)^{i} \cdot \idc{\*x[j:j]=1} = (1-p)K_{1,\*x}[i].
\end{align*}

As noted in \cite{mazooji2022substring}, the techniques of \cite{chen2020polynomialtime} in the smoothed complexity model of trace reconstruction can also be viewed as based on $k$-mer density maps. Indeed, for a fixed $w\in \{0,1\}^k$, the number of its occurrences as a subword in $\*x$ is $\sum_{j=0}^{n-1}\idc{\*x[j\colon j+k-1]=w} = \sum_{i=0}^{n-1}K_{w,\*x}[i]$. They show that for $k=O(\log n)$, the subword vector (indexed by $w\in \{0,1\}^k$) uniquely determines the source string, with high probability \cite[Lemma 1.1]{chen2020polynomialtime}.

The main result of \cite{mazooji2022substring} is that given access to $T=\eps^{-2}\cdot 2^{O(k)}\poly(n)$ traces of $\*x$, one can recover an estimation $\hat{K}_{\*x}$ of the $k$-mer density map $K_{\*x}$ which is entry-wise $\eps$-accurate, i.e., $\norm{\hat{K}_{\*x}-K_{\*x}}_{\ell_{\infty}} \le \eps$. We remark that by replacing $\eps$ with $\eps/(2^k n)$, one gets an estimate which is $\eps$-accurate in $\ell_1$-norm, while using asymptotically the same number of traces.

We make the following definition generalizing mean-based algorithms (\cite{de2019optimal, nazarov2017trace}).

\begin{definition}(Algorithms based on $k$-mer statistics) \label{def:k-mer-alg}
A trace reconstruction algorithm based on $k$-mer statistics works in two steps as follows:
\begin{enumerate}
\item Once the unknown source string $\*x\in \{0,1\}^n$ is picked, it chooses an accuracy parameter $\eps \in (0,1]$. It then receives an $\epsilon$-accurate estimate (in $\ell_1$-norm) of the $k$-mer density map $K_{\*x}$ based on the traces. From here on the algorithm has no more access to the traces themselves.   We define the cost of this part to be $2^k/\eps$. 

\item The algorithm may perform further post-processing and finish by outputting the source string.
\end{enumerate}
   \end{definition}
Since there is an algorithm to $\eps$-estimate the $k$-mer density map with $\eps^{-2}\cdot 2^{O(k)}\poly(n)$ many traces \cite{mazooji2022substring}, it follows that an algorithm defined as in Definition \ref{def:k-mer-alg}  with cost $T$ can be turned into a trace reconstruction algorithm with $\poly(T)$ samples.

We note that the $k$-mer density map estimators of \cite{mazooji2022substring} only use $k$-bit statistics of the traces, in fact statistics about contiguous $k$ bits in the traces, and hence $k$-mer-based algorithms are a subclass of 
algorithms based on $k$-bit statistics.

In this work,  we first observe that the upper bounds of Chase \cite{chase2021separating} can be in fact obtained via $k$-mer-based algorithms (see the formal statement in \thmref{thm:chaseub}), and hence by only using statistics of contiguous subwords of the traces. Our main result says that $k$-mer-based algorithms require $\exp(\Omega(n^{1/5})\sqrt{n})$ many traces (see \thmref{thm:main}). In addition, the analysis of this result implies that the proof technique in Chase \cite{chase2021separating} cannot lead to a better analysis of the sample complexity (up to $\log^{4.5} n$ factors in the exponent), and hence new techniques are needed to significantly improve the current upper bound.

\paragraph{The Maximum Likelihood Estimator} In model estimation settings, a common tool for picking a ``model'' that best explains the observed data is the Maximum Likelihood Estimator (MLE). In the setting of trace reconstruction, it is natural to ask: What is the most likely trace distribution $\+D_{\*x}$ (and hence $\*x$) to have produced the given sample/trace(s)? We formalize MLE next. 
\begin{definition}[Maximum Likelihood Estimation] \label{def:mle}
	Let $\+D=\set{D_1, D_2,\dots, D_m}$ be a finite set of probability distributions over a common domain $\Omega$. Given a sample $x \in \Omega$, the output of the Maximum Likelihood Estimation is (ties are broken arbitrarily)
	\begin{align*}
		\mle(x; \+D) \coloneqq \arg \max_{i \in [m]} D_i(x).
	\end{align*}
 For independently and identically distributed samples $x_1, x_2, \ldots, x_k \in \Omega$ the output of the Maximum Likelihood Estimation is (ties are broken arbitrarily) is 
	\begin{align*}
		\mle(x_1, x_2, \ldots x_k; \+D) \coloneqq \arg \max_{i \in [m]} \prod_{j \in [k]} D_i(x_j).
	\end{align*}
\end{definition}

We present a simple proof that this algorithm (which takes exponential time, as it searches through all $\*x\in \{0,1\}^n$) is in fact optimal in the number of {\em traces} used, up to an $O(n)$ factor blowup.

We also observe that in the average-case setting, where the source string is a uniformly random string from $\{0,1\}^n$, $\mle$ is indeed optimal -- without the $O(n)$ factor blowup (see Remark \ref{remark:avg-mle}.)

\subsection{Our Contributions}

\paragraph{The power of $k$-mer-based algorithms} Our first result shows that algorithms based on $k$-mer statistics can reconstruct a source string using $\exp(\widetilde{O}(n^{1/5}))$ many traces. This follows from the following theorem.

\begin{theorem}[Implied by \cite{chase2021separating}] \label{thm:chaseub}
Let $\*x, \*y \in \set{0,1}^n$ be two arbitrary distinct strings, and let $K_{\*x}, K_{\*y}$ be their $k$-mer density maps, respectively. Assuming $k=2n^{1/5}$, it holds that
\begin{align*}
	\norm{K_{\*x}-K_{\*y}}_{\ell_1} \ge \exp\tp{-O(n^{1/5}\log^5 n)}.
\end{align*}
\end{theorem} 

Based on \thmref{thm:chaseub}, the algorithm estimates $\hat{K}$ within an accuracy of $\eps=\exp(-O(n^{1/5}\log^5 n))$ and outputs the $\*x$ that minimizes $\norm{{\hat K}-K_{\*x}}_{\ell_1}.$ The cost of this $k$-mer-based algorithm is $\exp(O(n^{1/5}\log^5n))$.

Our main result regarding $k$-mer-based algorithms is the following theorem which shows the tightness of the bound in \thmref{thm:chaseub}.

\begin{theorem} \label{thm:main}
Fix any $k\le n^{1/5}$. Suppose $K_{\*x}$ stands for the $k$-mer density map of $\*x$. There exist distinct strings $\*x, \*y \in \set{0,1}^n$ such that
\begin{align*}
	\norm{K_{\*x}-K_{\*y}}_{\ell_1} \le \exp\tp{-\Omega(n^{1/5}\sqrt{\log n})}.
\end{align*}
\end{theorem}

Hence, \thmref{thm:main} implies that the cost of any $k$-mer-based algorithm for worst-case trace reconstruction is $\exp(\Omega(n^{1/5}\sqrt{\log n}))$.

\begin{remark}
As one might expect, for $k'<k$ the $k'$-mers usually contain less information than $k$-mers. To see this, observe that for a $(k-1)$-mer $w$, we have the following relation
\begin{align*}
    \idc{\*x[j:j+k-2]=w} = \idc{\*x[j-1:j+k-2]=0w} + \idc{\*x[j-1:j+k-2]=1w},
\end{align*}
provided that $j>0$. The same also holds for $\*y$. In fact, the strings $\*x$ and $\*y$ obtained via \cref{thm:main} share a common prefix of length at least $k$ (or one could prepend a prefix anyway), so $\*x[0:k'-1]=\*y[0:k'-1]$ for any $k'<k$, and one does not need to worry about the case $j=0$. Plugging into the definition of $k$-mer density maps, we have
\begin{align*}
    K_{w,\*x}[i]-K_{w,\*y}[i] = \tp{K_{0w,\*x}[i] - K_{0w,\*y}[i]}+\tp{K_{1w,\*x}[i] - K_{1w,\*y}[i]}.
\end{align*}
By induction, for any $k'<k$ we have
\begin{align*}
    \sum_{w\in\set{0,1}^{k'}}\abs{K_{w,\*x}[i]-K_{w,\*y}[i]} \le \sum_{w\in\set{0,1}^{k'}}\sum_{u \in \set{0,1}^{k-k'}}\abs{K_{uw,\*x}[i]-K_{uw,\*y}[i]} = \norm{K_{\*x}-K_{\*y}}_{\ell_1}.
\end{align*}
Therefore, the bound in \cref{thm:main} indeed covers all $k'$-mers for $k'\le k$.
\end{remark}

We remark that the proof of \thmref{thm:main} further implies that the analysis technique of \cite{chase2021separating} is essentially tight, in the sense that no better upper bound (up to $\log^{4.5}n$ factors in the exponent) can be obtained via his analysis. We include further details about this implication in Remark \ref{remark:analysis-tighness}. 

\paragraph{Maximum Likelihood Estimator: an optimal algorithm} We next turn to analyzing the performance of the MLE algorithm in the setting of trace reconstruction. Our main result essentially shows that if there is an algorithm for trace reconstruction that uses $T$ traces and succeeds with probability $3/4$ then the MLE algorithm using $O(nT)$ traces succeeds with probability $3/4.$ Hence, given that the current upper bounds for the worst-case reconstruction problem are exponential in $n$, we may view the MLE as an optimal algorithm for trace reconstruction.

\begin{theorem} \label{thm:mle-optimality}
	Suppose $\+D=\set{D_0, D_1, \dots, D_m}$ is such that $\tv{D_0, D_i} \ge 1-\eps$ for any $1\le i \le m$. Then we have
	\begin{align*}
		\PRs{x \sim D_0}{\mle(x;\+D)=0} \ge 1-m\eps.
	\end{align*}
\end{theorem}

We remark that the loss of a factor of $m$ in \cref{thm:mle-optimality} is generally inevitable. Here is a simple example: let $D_0$ be the uniform distribution over $[m]$, and for $i=1,2,\dots,m$, let $D_i$ be the point distribution supported on $\set{i}$. We have $\tv{D_0, D_i}=((m-1)/m + (1-1/m))/2=1-1/m$. However, $\PRs{x\sim D_0}{\mle(x;\+D)=0}=0$.

For a string $\*x \in \set{0,1}^n$, let $D_{\*x}$ denote the trace distribution of $\*x$. \cref{thm:mle-optimality} implies the following corollary, which implies that in some sense the Maximum Likelihood Estimation is a universal algorithm for trace reconstruction.

\begin{corollary} \label{cor:mle-trace}
    Suppose $T$ traces are sufficient for worst-case trace reconstruction with a success rate $3/4$. Then for any $\eps > 0$, Maximum Likelihood Estimation with $8\ln(1/\eps)\cdot nT$ traces solves worst-case trace reconstruction with success rate $1-\eps$.
\end{corollary}

\cref{cor:mle-trace} incurs a factor of $O(n)$ to the sample complexity. While we currently do not know whether this blowup is necessary for trace reconstruction, the next result shows that it is inevitable for the more general ``model estimation'' problem.  

\begin{theorem}\label{thm:mle-lb}
For any integer $n\ge 1$, there is a set of distributions $\+D=\set{D_0, D_1, D_2, \dots, D_m}$ over a common domain $\Omega$ of size $\abs{\Omega} = m + n$, where $m=\binom{n}{\floor{n/4}}=2^{\Theta(n)}$, satisfying the following conditions.
\begin{enumerate}
 \item There is a distinguisher $A$ which given \emph{one} sample $x\sim A_j$ for an unknown $j \in \set{0,1,\dots,m}$, recovers $j$ with probability at least $2/3$. In other words, for all $j=0,1,\dots,m$, 
 \begin{align*}
    \Pr_{x\sim D_j}[A(x)=j] \ge 2/3.
 \end{align*}
 \item $\mle$ fails to distinguish $D_0$ from other distributions with probability 1, even with $T=n/4$ samples. In other words, 
 \begin{align*}
    \Pr_{x_1,\dots,x_T\sim D_0}[\mle(x_1,\dots,x_T;\+D)=0]=0.
 \end{align*}
\end{enumerate}
\end{theorem}

\begin{remark}\label{remark:avg-mle}
Finally, we remark that in the average-case setting $\mle$ is indeed optimal (with no factor of $O(n)$ factor blowup in the number of traces). This is because maximizing the likelihood is equivalent to maximizing the posterior probability under the uniform prior distribution (which is optimal), as can be seen via the Bayes rule 
\begin{align*}
    \+D_{\*x}(\tx_1, \dots, \tx_{T}) =& p(\*x \mid \tx_1, \dots, \tx_{T}) \cdot \frac{\sum_{\*x' \in \set{0,1}^n} p(\*x')\cdot \+D_{\*x'}(\tx_1, \dots, \tx_{T})}{p(\*x)} \\
    =& p(\*x \mid \tx_1, \dots, \tx_{T}) \cdot \sum_{\*x'\in\set{0,1}^n}\+D_{\*x'}(\tx_1, \dots, \tx_{T}) \\
    =& p(\*x \mid \tx_1, \dots, \tx_{T}) \cdot f(\tx_1, \dots, \tx_{T}).
\end{align*}
Therefore maximizing both sides with respect to $\*x$ yields the same result.
\end{remark}

\subsection{Overview of the techniques} 
\paragraph{Lower bounds for $k$-mer-based algorithms}
In recent development of the trace reconstruction problem, the connection to various real and complex polynomials has been a recurring and intriguing theme \cite{HolensteinMPW08, nazarov2017trace, PeresZ17, HoldenPP18, de2019optimal, chen2020polynomialtime, cheraghchi2021mean, chase2021separating, sima2021trace, grigorescu2022limitations, rubinstein2023average}. The starting point of these techniques is to design a set of statistics that can be easily estimated from the traces (\eg, mean traces), with the property that for different source strings the corresponding statistics are somewhat ``far apart''. To establish this property, one key idea is to associate each source string $\*x$ with a generating polynomial $P_{\*x}$ where the coefficients are exactly the statistics of $\*x$. Due to the structure of the deletion channel, in many cases, this generating polynomial (under a change of variables) is identical to another polynomial $Q_{\*x}$ that is much easier to get a handle on. For example, the coefficients of $Q_{\*x}$ are usually 0/1, and they are easily determined from $\*x$. To show that the statistics corresponding to $\*x$ and $\*y$ are far apart (say, in $\ell_1$-distance), it is sufficient to show that $\abs{Q_{\*x}(w)-Q_{\*y}(w)}$ is large for an appropriate choice of $w$. This is the point where all sorts of analytical tools are ready to shine. For instance, the main technical result in \cite{chase2021separating} is a complex analytical result that says that a certain family of polynomials cannot be uniformly small over a sub-arc of the complex unit circle, which has applications beyond the trace reconstruction problem.

This analytical view of trace reconstruction can lead to a tight analysis of certain algorithms/statistics. The best example would be mean-based algorithms, for which a tight bound of $\exp(\Theta(n^{1/3}))$ traces is known to be sufficient and necessary for worst-case trace reconstruction \cite{nazarov2017trace, de2019optimal}. The tightness of the sample complexity is exactly due to the tightness of a complex analytical result by Borwein and Erd\'elyi \cite{borwein1997littlewood}. Our lower bound for $k$-mer-based algorithms is obtained in a similar fashion, via establishing a complex analytical result complementary to that of \cite{chase2021separating} (See  \cref{lem:complex-main}).

On the other hand, our argument takes a different approach than that of \cite{borwein1997littlewood}. At a high level, both results use a Pigeonhole argument to show the existence of two univariate polynomials which are uniformly close over a sub-arc $\Gamma$ of the complex unit circle. The difference lies in the objects playing the role of ``pigeons''. \cite{borwein1997littlewood}'s argument can be viewed as two steps: (1) apply the Pigeonhole Principle to obtain two polynomials that have close evaluations over a \emph{discrete} set of points in $\Gamma$, and (2) use a continuity argument to extend the closeness to the entire sub-arc. Here the roles of pigeons and holes are played by evaluation vectors, and Cartesian products of small intervals. Our approach considers the coordinates of a related polynomial in the \emph{Chebyshev basis}, which play the roles of pigeons in place of the evaluation vector. The properties of Chebyshev polynomials allow us to get rid of the continuity argument. Instead, we complete the proof by leveraging rather standard tools from complex analysis (\eg, \cref{thm:chebcoef-bound} and \cref{thm:hadamard3circle}). We believe this approach has the advantage of being generalizable to multivariate polynomials over the product of sub-arcs $\Gamma = \Gamma_1 \times \dots \times \Gamma_m$ via multivariate Chebyshev series (see, \eg, \cite{mason1980near, trefethen2017multivariate}), whereas the same generalization seems to be tricky for the continuity argument. 

Finally, the counting argument considers a special set of strings for which effectively only one $k$-mer contains meaningful information about the initial string. Since previous arguments did not exploit structural properties of the strings, this is another technical novelty of our proof.

\paragraph{Maximum Likelihood Estimation} 
Most of our results regarding Maximum Likelihood Estimation hold under the more general ``model estimation'' setting, where one is given a sample $x$ drawn from an unknown distribution $D\in \+D$ and tries to recover $D$. Our main observation is that if such a distinguisher works in worst-case, then the distributions in $\+D$ have large pairwise statistical distances. The maximization characterization of statistical distance, in conjunction with a union bound, implies that for a sample $x\sim D$ its likelihood is maximized by $D$ except with a small probability. The $O(n)$ factor loss in the sample complexity is essentially due to the union bound, and we show that this loss is tight in general by constructing a set of distributions which attains equality in the union bound.

\subsection{Related work}
The trace reconstruction problem was first introduced and studied by Levenshtein \cite{Levenshtein01a}\cite{Levenshtein01b}. The original question  is that if a message is sent multiple times through the same channel with random insertion/deletion errors, then how to recover the message?
\cite{BatuKKM04} and \cite{HolensteinMPW08} formalized the problem to the current version for which the channel only has random deletions. 
Their central motivation is actually from computational biology, i.e. how to reconstruct the whole DNA sequence from multiple related subsequences. 
\cite{cheraghchi2020coded} and \cite{brakensiek2020coded} further extended the study to the ``coded'' version. 
That is, the string to reconstruct is not an arbitrary string but instead is a codeword from a code.
A variant setting where the channel has memoryless replication insertions was studied by \cite{cheraghchi2021mean}.

The {\em average} case version was studied in \cite{HolensteinMPW08, PeresZ17, McGregorPV14, HoldenPP18}. 
For this case, the best known lower bound on the number of traces is $\widetilde{\Omega}( \log^{5/2} n )$ \cite{HoldenL18, chase2021new}.
Building on Chase's upper bound for the worst case, \cite{rubinstein2023average} improved the sample complexity upper bound to $\exp(\widetilde{O}(\log^{1/5}n))$ in the average-case model.

\cite{chen2020polynomialtime} studied another variant of the problem which is called the smooth variant.
It is an intermediate model between the worst-case and the average-case models, where the initial string is an arbitrary string perturbed by replacing each coordinate by a
uniformly random bit with some constant probability in $[0, 1]$. 
\cite{chen2020polynomialtime} provided an efficient reconstruction algorithm for this case.
Other variants studied include trace reconstruction from the multiset of substrings \cite{GabrysM17, GabrysM19},
population recovery variants \cite{FanCFSS19}, matrix reconstruction and parameterized algorithms \cite{KrishnamurthyM019}, circular trace reconstruction \cite{narayanan2021circular}, reconstruction from $k$-decks \cite{KrasikovR97, scott1997reconstructing,dudik2003reconstruction,McGregorPV14}, and coded trace reconstruction\cite{cheraghchi2020coded,brakensiek2020coded}.

\cite{davies2021approximate} studied {\em approximate} trace reconstruction and showed efficient algorithms.
\cite{chen2021near}, \cite{chakraborty2021approximate}, and \cite{chase2021approximate} further proved that if the source is a random string, then an approximate solution can be found with high probability using very few traces.
Notice that  approximate reconstructions imply distinguishers for pairs of
strings with large edit distances. 
\cite{McGregorPV14, sima2021trace, grigorescu2022limitations} study the complexity of the problem parameterized by the Hamming/edit distance between the strings. \cite{grigorescu2022limitations} also shows that the problem of exhibiting explicit strings that are hard to distinguish for mean-based algorithms is equivalent to the Prouhet-Tarry-Escott problem, a difficult problem in number theory.

\subsection{Organization} In Section \ref{sec:ub} we prove \thmref{thm:chaseub}, in Section \ref{sec:lb} we prove our main result \thmref{thm:main}, and in Section \ref{sec:mle} we prove \thmref{thm:mle-optimality}.

\section{$k$-mer-based algorithms: the upper bound}\label{sec:ub}
\newcommand{\tE}{\widetilde{E}}
We prove \cref{thm:chaseub} in this section.

Let us start with a definition that is essential for the study of $k$-mer-based algorithms.  
\begin{definition}[$k$-mer generating polynomial] \label{def:kmer-genpoly}
Let $\*x \in \set{0,1}^n$ and $w \in \set{0,1}^k$. The \emph{$k$-mer generating polynomial} $P_{w,\*x}$ for string $\*x$ and $k$-mer $w$ is the following degree-$(n-1)$ polynomial in $z$:
\begin{align*}
	P_{w,\*x}(z) \coloneqq \sum_{\ell=0}^{n-1}K_{w,\*x}[\ell]\cdot z^{\ell}. 
\end{align*}
\end{definition}
We have the following identity
\begin{align*}
	P_{w,\*x}(z) &= \sum_{\ell=0}^{n-1}K_{w,\*x}[\ell]\cdot z^{\ell} \\
	&= \sum_{\ell=0}^{n-1}\tp{ \sum_{j=0}^{n-k}\binom{j}{\ell}(1-p)^{\ell}p^{j-\ell} \cdot \idc{\*x[j\colon j+k-1]=w}} z^{\ell} \\
	&= \sum_{j=0}^{n-k}\idc{\*x[j\colon j+k-1]=w} \cdot \tp{p+(1-p)z}^{j}.
\end{align*}
The expression on the last line, under a change of variable $z_0=p+(1-p)z$, is exactly the polynomial studied in \cite{chase2021separating}. 

\begin{lemma}{\cite[Proposition 6.3]{chase2021separating}} \label{lem:chase-prop}
For distinct $\*x, \*y \in \set{0,1}^n$, if $x_i = y_i$ for all $0\le i< 2n^{1/5}-1$, then there are $w \in \set{0,1}^{2n^{1/5}}$ and $z_0 \in \set{e^{i\theta}\colon \abs{\theta} \le n^{-2/5}}$ such that
\begin{align*}
	\abs{ \sum_{j\ge 0}\tp{\idc{\*x[j\colon j+2n^{1/5}-1]=w}-\idc{\*y[j\colon j+2n^{1/5}-1]=w}} z_0^{j} } \ge \exp\tp{-Cn^{1/5}\log^5 n}.
\end{align*}
Here $C>0$ is a constant depending only on the deletion probability $p$.
\end{lemma}

We will use \cref{lem:chase-prop} to show that the $\exp(\widetilde{O}(n^{1/5}))$ upper bound of \cite{chase2021separating} can be achieved by $k$-mer-based algorithms, rather than general algorithms based on $k$-bit statistics. Our main lower bound on the number of traces implied by \thmref{thm:main} will follow by showing an upper bound on the LHS in the lemma above (see Lemma \ref{lem:complex-main}).

\begin{remark}\label{remark:analysis-tighness}
    We remark that the result of Chase is obtained by first considering a corresponding multivariate channel polynomial that encodes in its coefficients the $k$-bit statistics of the traces.   
    The upper bound on the number of traces reduces to understanding the supremum of this polynomial over a certain region of the complex plane.   
    The crucial element of the proof is the reduction to the existence of $w\in \{0,1\}^k$ and $z_0$ satisfying Lemma \ref{lem:chase-prop}, by appropriately making the remaining variables take value $0$. We noticed that the resulting univariate polynomial is essentially the $k$-mer generating polynomial defined in \cref{def:kmer-genpoly}, with an extra factor of $(1-p)^{k}$.   
    Our result in Lemma \ref{lem:complex-main} implies that no tighter lower bound  (up to polylogarithmic factors in the exponent) is possible for this univariate polynomial, showing that the analysis technique used in \cite{chase2021separating} cannot give a better upper bound on worst-case trace complexity.

\end{remark}

\subsection{An upper bound for $k$-mer based algorithms}

The proof of \cref{thm:chaseub} mainly uses \cref{lem:chase-prop}. We will also make use of the following result.
\begin{lemma}{\cite[Theorem 5.1]{borwein1999littlewood}} \label{lem:BEK99}
There are absolute constants $c_1>0$ and $c_2>0$ such that
\begin{align*}
    \abs{f(0)}^{c_1/a} \le \exp\tp{\frac{c_2}{a}}\sup_{t \in [1-a,1]}\abs{f(t)}
\end{align*}
for every analytic function $f$ on the open unit disk that satisfies $\abs{f(z)}<(1-\abs{z})^{-1}$ for $\abs{z}<1$, and $a\in (0,1]$.
\end{lemma}

\begin{proof}[Proof of \cref{thm:chaseub}]
	
\phantom\newline

The proof deals with two cases.

\textbf{Case 1:} $x_i = y_i$ for all $0\le i<2n^{1/5}-1$.

In this case, $\*x$ and $\*y$ satisfy the premise of \cref{lem:chase-prop}. It follows that there exist $w \in \set{0,1}^{2n^{1/5}}$, and $z_0 =e^{i\theta}$ where $\abs{\theta}\le n^{-2/5}$, satisfying the bound
\begin{align*}
	\abs{ \sum_{j\ge 0}\tp{\idc{\*x[j\colon j+2n^{1/5}-1]=w}-\idc{\*y[j\colon j+2n^{1/5}-1]=w}} z_0^{j} } \ge \exp\tp{-Cn^{1/5}\log^5 n}.
\end{align*}
Here $C > 0$ is a constant depending only on the deletion probability $p$. Rewriting in terms of the $k$-mer generating polynomials, we have
\begin{align} \label{eqn:objcase1}
\abs{P_{w,\*x}\tp{\frac{z_0-p}{1-p}}-P_{w,\*y}\tp{\frac{z_0-p}{1-p}}} \ge \exp\tp{-Cn^{1/5}\log^5 n}.
\end{align}
It is easy to see that $\abs{z_0-p}/\abs{1-p} \ge \abs{\abs{z_0}-p}/\abs{1-p}=1$. We also have the following upper bounds
\begin{align*}
	\abs{\frac{z_0-p}{1-p}}^2 &= \frac{(\cos\theta - p)^2 + \sin^2\theta}{(1-p)^2} = \frac{1-2p\cos\theta+p^2}{(1-p)^2} = 1 + \frac{2p(1-\cos\theta)}{(1-p)^2} \\
	&= 1 + \frac{4p\sin^2\frac{\theta}{2}}{(1-p)^2} \le 1 + \frac{p\theta^2}{(1-p)^2} \le 1 + \frac{p}{(1-p)^2}\cdot n^{-4/5}, \\
	\abs{\frac{z_0-p}{1-p}}^n &\le \tp{1 + \frac{p}{(1-p)^2}\cdot n^{-4/5}}^{n/2} \le \exp\tp{\frac{p}{(1-p)^2}\cdot n^{-4/5}\cdot \frac{n}{2}} \\
	&= \exp\tp{\frac{p}{2(1-p)^2}\cdot n^{1/5}}.
\end{align*}
From here we can apply the triangle inequality and conclude that
\begin{align*}
	\norm{K_{\*x}-K_{\*y}}_{\ell_1} &\ge \sum_{\ell=0}^{n-1}\abs{K_{w,\*x}[\ell]-K_{w,\*y}[\ell]} \\
	&\ge \abs{\frac{z_0-p}{1-p}}^{-n}\cdot \abs{\sum_{\ell=0}^{n-1}\tp{K_{w,\*x}[\ell]-K_{w,\*y}[\ell]}\cdot\tp{\frac{z_0-p}{1-p}}^{\ell}} \\
	&\ge \exp\tp{-\frac{p}{2(1-p)^2}\cdot n^{1/5}-Cn^{1/5}\log^5 n} \\
	&\ge \exp\tp{-C'n^{1/5}\log^5 n}.
\end{align*}
Here $C'=p(1-p)^{-2}/2+C$ is a constant depending only on the deletion probability $p$.

\textbf{Case 2:} $x_i\neq y_i$ for some $0\le i < 2n^{1/5}-1$, \ie, $\*x[0\colon 2n^{1/5}-1]\neq \*y[0\colon 2n^{1/5}-1]$.

In this case, we are going to take $w=\*x[0\colon 2n^{1/5}-1]$ and show a much better bound
\begin{align} \label{eqn:objcase2}
	\sup_{z \colon \abs{z}\le 1}\abs{P_{w,\*x}(z) - P_{w,\*y}(z)} > C'',
\end{align}
where $C''>0$ is a constant depending only on $p$ (hence certainly greater than $\exp(-\widetilde{O}(n^{1/5}))$). Similar to what we did in case 1, applying the triangle inequality to \cref{eqn:objcase2} gives the theorem.

To prove \cref{eqn:objcase2}, we let
\begin{align*}
	Q(z_0) = \sum_{j\ge 0}\tp{\idc{\*x[j\colon j+2n^{1/5}-1]=w}-\idc{\*y[j\colon j+2n^{1/5}-1]=w}} z_0^{j}, 
\end{align*}  
so that $Q(p+(1-p)z)=P_{w,\*x}(z)-P_{w,\*y}(z)$. Under our choice of $w$, the constant term of $Q$ equals to 1, \ie, $Q(0)=1$. 

If $p\in (0,1/2]$, the closed disk $B(p;1-p)=\set{p+(1-p)z\colon \abs{z}\le 1}$ contains the point 0. Therefore
\begin{align*}
\sup_{z\colon\abs{z}\le 1}\abs{P_{w,\*x}(z)-P_{w,\*y}(z)} = \sup_{z_0 \in B(p;1-p)}\abs{Q(z_0)} \ge \abs{Q(0)} = 1. 
\end{align*}
We are left with the case $p \in (1/2,1)$. Since $Q$ is a polynomial with coefficients absolutely bounded by 1, we can apply \cref{lem:BEK99} with $a=2(1-p) \in (0,1)$ and obtain
\begin{align*}
	\sup_{t_0 \in [1-a,1]}\abs{Q(t_0)} \ge \exp\tp{-\frac{c_1}{a}}\cdot \abs{Q(0)}^{c_2/a} = \exp\tp{-\frac{c_1}{a}}
\end{align*}
for constants $c_1, c_2 > 0$. Denoting $t = (t_0-p)/(1-p)$, we have $t \in [-1,1]$ when $t_0 \in [1-a,1]$. In particular, $t$ is inside the closed unit disk $B(0;1)$. Therefore 
\begin{align*}
\sup_{z\colon\abs{z}\le 1}\abs{P_{w,\*x}(z)-P_{w,\*y}(z)} \ge \sup_{t\in [-1,1]}\abs{P_{w,\*x}(t)-P_{w,\*y}(t)} = \sup_{t_0\in [1-a,1]}\abs{Q(t_0)} \ge \exp\tp{-\frac{c_1}{a}}.
\end{align*}
To conclude, we can take $C''=\min\set{1, \exp(-c_1(1-p)^{-1}/2)}$.
\end{proof}

\section{A lower bound for $k$-mer based algorithms: Proof of \cref{thm:main}}\label{sec:lb}
We prove \cref{thm:main} in this section. The proof is based on the following lemma, which we will prove shortly.
\begin{lemma} \label{lem:complex-main}
There exists $\*x, \*y \set{0,1}^n$ such that for any $k$-mer $w$, it holds that
\begin{align*}
	\sup_{z\colon\abs{z}=1}\abs{P_{w,\*x}(z)-P_{w,\*y}(z)} \le 2^{-cn^{1/5}\sqrt{\log n}}.
\end{align*}
\end{lemma}

\begin{proof}[Proof of \cref{thm:main} using \cref{lem:complex-main}]
We can extract $K_{w,\*x}[\ell]-K_{w,\*y}[\ell]$ by the contour integral (\cf \cite[\S 4, Theorem 2.1]{lang2013complex})
\begin{align*}
	K_{w,\*x}[\ell]-K_{w,\*y}[\ell] = \frac{1}{2\pi i}\int_{\abs{z}=1}\tp{P_{w,\*x}(z)-P_{w,\*y}(z)}\cdot z^{-\ell-1} \dif z.
\end{align*}
Therefore
\begin{align*}
	\abs{K_{w,\*x}[\ell]-K_{w,\*y}[\ell]} \le \frac{1}{2\pi}\int_{\abs{z}=1}\abs{P_{w,\*x}(z)-P_{w,\*y}(z)}\cdot \abs{z}^{-\ell-1} \cdot \abs{\dif z} \le 2^{-cn^{1/5}\sqrt{\log n}}.
\end{align*}
We stress that the bound holds for any $\ell \in [n]$ and $k$-mer $w$. Note that for any fixed $\ell$, there are at most $n-k+1$ different $k$-mers $w$ for which $K_{w,\*x}[\ell] > 0$. Namely, if $w \notin \set{x[j\colon j+k-1]\colon 0\le j\le n-k}$ then $K_{w,\*x}[\ell] = 0$. It follows that
\begin{align*}
	\norm{K_{\*x}-K_{\*y}}_{\ell_1} = \sum_{\ell=0}^{n-1}\sum_{w}\abs{K_{w,\*x}[\ell]-K_{w,\*y}[\ell]} \le n\cdot 2(n-k+1)\cdot 2^{-cn^{1/5}\log^{2/5}n} \le 2^{-c'n^{1/5}\sqrt{\log n}}.
\end{align*}
\end{proof}

Next, we prove \cref{lem:complex-main} assuming the following result, which is our main technical lemma.
\begin{lemma} \label{lem:technical-main}
Fix any $k \le L^{1/3}$. There exist distinct $\*x, \*y \in \set{0,1}^{L}$ both starting with a run of 0s of length $L^{1/3}-1$, such that for any $k$-mer $w$, it holds that
\begin{align*}
	\sup_{\theta\colon \abs{\theta}\le L^{-2/3}\log^{1/4}L}\abs{P_{w,\*x}(e^{i\theta})-P_{w,\*y}(e^{i\theta})}\le 2^{-L^{1/3}\sqrt{\log L}/20}.
\end{align*}
\end{lemma}

\begin{proof}[Proof of \cref{lem:complex-main} using \cref{lem:technical-main}]
Let $\beta \ge 3/5$ be a parameter to be decided later. Denote $L\coloneqq n^{\beta}$. We have $k \le n^{1/5} = L^{1/(5\beta)} \le L^{1/3}$, so that the  premise of \cref{lem:technical-main} is satisfied. Therefore, there exist distinct $\*x', \*y' \in \set{0,1}^{L}$ both starting with a run of 0s of length $L^{1/3}-1$, such that for any $k$-mer $w$, it holds that
\begin{align} \label{eqn:subarc-bound}
\sup_{\theta\colon \abs{\theta}\le L^{-2/3}\log^{1/4}L}\abs{P_{w,\*x'}(e^{i\theta})-P_{w,\*y'}(e^{i\theta})}\le 2^{-L^{1/3}\sqrt{\log L}/20}.
\end{align}
Let $\*x = 0^{n-L}\*x'$ and $\*y = 0^{n-L}\*y'$. Since $k\le L^{1/3}$, by construction we have $\*x[j:j+k-1]=\*y[j:j+k-1]$ for all $0\le j\le n-L$. Therefore, any $k$-mer $w$ we have 
\begin{align*}
	& P_{w,\*x}(e^{i\theta}) - P_{w,\*y}(e^{i\theta}) \\
    =& \sum_{j=0}^{n-k}\tp{\idc{\*x[j:j-k+1]=w}-\idc{\*y[j:j-k+1]=w}}(p+qe^{i\theta})^j \\
	=& \tp{p+qe^{i\theta}}^{n-L} \cdot \sum_{j=n-L}^{n-k}\tp{\idc{\*x[j:j-k+1]=w}-\idc{\*y[j:j-k+1]=w}}(p+qe^{i\theta})^{j-(n-L)} \\
	=& \tp{p+qe^{i\theta}}^{n-L}\cdot \sum_{j=0}^{L-k}\tp{\idc{\*x'[j:j-k+1]=w}-\idc{\*y'[j:j-k+1]=w}}(p+qe^{i\theta})^j \\
	=& \tp{p+qe^{i\theta}}^{n-L} \cdot \tp{P_{w,\*x'}(e^{i\theta}) - P_{w,\*y'}(e^{i\theta})}.
\end{align*}
Here $q=1-p$. When $\abs{\theta}$ is large, we can upper bound the supremum as
\begin{align*}
	\sup_{\theta\colon \abs{\theta}> L^{-2/3}\log^{1/4}L}\abs{P_{w,\*x}(e^{i\theta}) - P_{w,\*y}(e^{i\theta})} &= \sup_{\theta\colon \abs{\theta}> L^{-2/3}\log^{1/4}L}\abs{p+qe^{i\theta}}^{n-L} \cdot \abs{P_{w,\*x'}(e^{i\theta}) - P_{w,\*y'}(e^{i\theta})} \\
	&\le \tp{1-c_1L^{-4/3}\log^{1/2}L}^{n-L} \cdot \sup_{\theta\colon \abs{\theta}> L^{-2/3}\log^{1/4}L} \abs{P_{w,\*x'}(e^{i\theta}) - P_{w,\*y'}(e^{i\theta})} \\
	&\le \exp\tp{-c_1(n-L)L^{-4/3}\log^{1/2}L} \cdot (L-k+1) \\
	&\le \exp_2\tp{-c_2n^{1-4\beta/3}\log^{1/2}n}.
\end{align*}
Here the first inequality is due to $\abs{p+qe^{i\theta}} \le 1-c_1a^2$ for some constant $c_1$ (depending on $p$) when $\abs{\theta}\ge a$. When $\abs{\theta}$ is small, this is taken care of by \cref{eqn:subarc-bound}:
\begin{align*}
	\sup_{\theta\colon \abs{\theta}\le L^{-2/3}\log^{1/4}L}\abs{P_{w,\*x}(e^{i\theta}) - P_{w,\*y}(e^{i\theta})} &\le \sup_{\theta\colon \abs{\theta}\le L^{-2/3}\log^{1/4}L} \abs{P_{w,\*x'}(e^{i\theta}) - P_{w,\*y'}(e^{i\theta})} \\
	&\le \exp_2\tp{-L^{1/3}\sqrt{\log L}/20} \\
 &\le \exp_2\tp{-c_3 n^{\beta/3}\log^{1/2}n}.
\end{align*}
Finally, the value of $\beta$ is determined by balancing the two cases. Namely, we let $1-4\beta/3=\beta/3$, or $\beta = 3/5$, which gives the bound $2^{-cn^{1/5}\sqrt{\log n}}$ for both cases. Here $c=\min\set{c_2, c_3}$.
\end{proof}

It remains to prove \cref{lem:technical-main}, which we do after some helpful preliminaries from complex analysis.

\subsection{Some helpful results in complex analysis}
In this section, we introduce some results in complex analysis, which will be useful for proving \cref{lem:technical-main}.

Let $T_d(x)$ denote the $d$\th Chebyshev polynomial, i.e., a degree-$d$ polynomial such that $T_d(\cos\theta)=\cos(d\theta)$. Clearly, $T_d(x) \in [-1,1]$ for $x\in [-1,1]$. If a function $f(z)$ is analytic on $[-1,1]$, it has a converging Chebyshev expansion
\begin{align*}
	f(z) = \sum_{d=0}^{\infty}a_d\cdot T_d(z),\quad z\in[-1,1].
\end{align*} 
Here the $a_d$'s are the \emph{Chebyshev coefficients}, and they can be extracted by the following integral
\begin{align*}
	a_d = \frac{1}{\pi}\int_{0}^{2\pi}f(\cos\theta)\cos(d\theta)\dif\theta,\quad d\ge 1,
\end{align*}
where $\pi$ is replaced by $2\pi$ for $d=0$. This immediately implies a uniform upper bound on Chebyshev coefficients.
\begin{proposition} \label{prop:trivial-chebcoef-bound}
	For all $d\ge 0$, $\abs{a_d} \le 2\sup_{x \in [-1,1]}\abs{f(x)}$.
\end{proposition}

In fact, if $f$ is analytically continuable to a larger region, much better bounds can be obtained. For that we need the notion of Bernstein ellipse.
\begin{definition}[Bernstein Ellipse]
	Given $\rho \ge 1$, the boundary of the Bernstein Ellipse is defined as
	\begin{align*}
		\partial E_{\rho} \coloneqq \Set{\frac{u+u^{-1}}{2} \colon u = \rho e^{i\theta}, \theta \in [0,2\pi) }.
	\end{align*}
\end{definition}
The Bernstein Ellipse $E_{\rho}$ has the foci at $\pm 1$ with the major and minor semi-axes given by $(\rho+\rho^{-1})/2$ and $(\rho-\rho^{-1})/2$, respectively. When $\rho = 1$, $E_{\rho}$ coincides with the interval $[-1,1]$ on the real line. For our purpose, we will also be working with affine transformations of $E_{\rho}$. More precisely, for $a \in [0,1/8]$ we denote by $\tE_{a,\rho}$ (the interior of) the  following ellipse
\begin{align*}
	\Set{(1-4a)+4a\cdot \frac{u+u^{-1}}{2} \colon u = \rho e^{i\theta}, \theta \in [0,2\pi)}.
\end{align*}
Thus, $\tE_{a,\rho}$ can be equivalently defined as 
\begin{align*}
	\Set{z \colon \abs{z-(1-8a)} + \abs{z-1} \le 8a + 4a(\rho-1)^2/\rho}.
\end{align*}

Below are some useful properties of $\tE_{a,\rho}$.

\begin{proposition} \label{prop:ellipse-bound}
	The following statements hold.
	\begin{enumerate}
		\item Let $z \in \partial\tE_{a,\rho}$. Then $\abs{z}\le 1+2a(\rho-1)^2/\rho$.
		\item $\tE_{a,\rho}$ contains a disk centered at 1 with radius $2a(\rho-1)^2/\rho$.
	\end{enumerate}
\end{proposition} 
\begin{proof}
	\textit{Item (1):}
	
	Writing $z=(1-4a)+4a(u+u^{-1})/2$ where $u=\rho e^{i\theta}$, we have
	\begin{align*}
		\abs{z}^2 &= \tp{1-4a+2a\rho\cos\theta+2a\rho^{-1}\cos\theta}^2+\tp{2a\rho\sin\theta-2a\rho^{-1}\sin\theta}^2 \\
		&= (1-4a)^2 + 2(1-4a)\cdot 2a(\rho+\rho^{-1})\cos\theta + \tp{2a(\rho+\rho^{-1})}^2\cos^2\theta + \tp{2a(\rho-\rho^{-1})}^2\sin^2\theta \\
		&\le (1-4a)^2 + 2(1-4a)\cdot 2a(\rho+\rho^{-1}) + \tp{2a(\rho+\rho^{-1})}^2 - 4a^2\sin^2\theta\tp{(\rho+\rho^{-1})^2-(\rho-\rho^{-1})^2} \\
		&= \tp{(1-4a)+2a(\rho+\rho^{-1})}^2 - 16a^2\sin^2\theta \\
		&\le \tp{1+2a(\rho+\rho^{-1}-2)}^2.
	\end{align*}
	Therefore $\abs{z} \le 1+2a(\rho+\rho^{-1}-2) = 1+2a(\rho-1)^2/\rho$.
	
	\textit{Item (2):}
	Let $z$ be such that $\abs{z-1}\le 2a(\rho-1)^2/\rho$. We have
	\begin{align*}
		\abs{z-(1-8a)}+\abs{z-1} \le 8a+2\abs{z-1} \le 8a+4a(\rho-1)^2/\rho.
	\end{align*}
	This implies $z \in \tE_{a,\rho}$.
\end{proof}

The following result shows an exponential convergence rate of the Chebyshev expansion.
\begin{theorem}[Theorem 8.1,~\cite{trefethen2012approximation}] \label{thm:chebcoef-bound}
	Let a function $f$ analytic on $[-1,1]$ be analytically continuable to the open Bernstein Ellipse $E_{\rho}$, where it satisfies $\abs{f(z)} \le M$ for some $M$. Then its Chebyshev coefficients satisfy 
	\begin{align*}
		\abs{a_0}\le M,\textup{ and }\abs{a_k} \le 2M\rho^{-k}, k\ge 1.
	\end{align*}
\end{theorem}
\begin{proof}
	The Chebyshev coefficients of $f$ is given by
	\begin{align*}
		a_k = \frac{1}{\pi}\int_{0}^{2\pi}f\tp{\cos\theta}T_k\tp{\cos\theta}\dif\theta = \frac{1}{\pi}\int_{0}^{2\pi}f\tp{\cos\theta}\cos\tp{k\theta}\dif\theta,
	\end{align*}
	with $\pi$ replaced by $2\pi$ for $k=0$. Letting $z = e^{i\theta}$, one could write $\cos\theta = (z+z^{-1})/2$, $\dif\theta = (iz)^{-1}\dif z$, and hence
	\begin{align*}
		a_k = \frac{1}{\pi i}\int_{\abs{z}=1}f\tp{\frac{z+z^{-1}}{2}}\frac{z^k+z^{-k}}{2}\cdot z^{-1}\dif z.
	\end{align*}
	Denote $F(z)\coloneqq f((z+z^{-1})/2)=F(z^{-1})$. Note that we can substitute $z^{-1}$ for $z$ and obtain
	\begin{align*}
		\frac{1}{\pi i}\int_{\abs{z}=1}F(z)z^{k-1}\dif z = -\frac{1}{\pi i}\int_{\abs{z}=1}F(z^{-1})z^{-(k-1)}\dif z^{-1}	= \frac{1}{\pi i}\int_{\abs{z}=1}F(z)z^{-k-1}\dif z.
	\end{align*}
	Therefore we arrived at the expression
	\begin{align*}
		a_k = \frac{1}{\pi i}\int_{|z|=1}F(z)z^{-k-1}\dif z.
	\end{align*}
	Since $f(z)$ is analytic in the open Bernstein Ellipse $E_{\rho}$, we can conclude that $F(z)$ is analytic in the annulus $\rho^{-1}<\abs{z}<\rho$. That means, for any $\rho_0 \in (\rho^{-1},\rho)$ we have by Cauchy's integral theorem (\cf \cite[\S 3, Theorem 5.1]{lang2013complex}) that
	\begin{align*}
		a_k = \frac{1}{\pi i}\int_{\abs{z}=\rho_0}F(z)z^{-k-1}\dif z.
	\end{align*}
	Now we have 
	\begin{align*}
		\abs{a_k} \le \frac{1}{\pi} \cdot \int_{\abs{z}=\rho_0} \abs{F(z)} \cdot \abs{z}^{-k-1} \abs{\dif z} \le \frac{1}{\pi} \cdot 2\pi \rho_0 M \cdot \rho_0^{-k-1} = 2M\rho_0^{-k}.
	\end{align*}
	Finally, since the bound holds for any $\rho_0<\rho$, it also holds for $\rho_0=\rho$.
\end{proof}

We will also make use of the following theorem.
\begin{theorem}[Hadamard Three Circles Theorem] \label{thm:hadamard3circle}
	Suppose $f$ is analytic inside and on $\set{z\in\mathbb{C}\colon r_1\le \abs{z}\le r_2}$. For $r\in [r_1, r_2]$, let $M(r)\coloneqq \sup_{\abs{z}=r}\abs{f(z)}$. Then
	\begin{align*}
		M(r)^{\log(r_2/r_1)} \le M(r_1)^{\log(r_2/r)}M(r_2)^{\log(r/r_1)}.
	\end{align*}
\end{theorem}

\begin{corollary} \label{cor:ellipse-to-interval}
	Suppose $f(z)=\sum_{j=0}^{n-1}c_jz^j$ where $\abs{c_j}\le 1$. Then
	\begin{align*}
		\sup_{z\in\partial\tE_{a,2}}\abs{f(z)} \le \exp\tp{5an/2} \cdot \tp{\sup_{z\in[1-8a,1]}\abs{f(z)} }^{1/2}.
	\end{align*}
\end{corollary}
\begin{proof}
	Let $\rho_1=1, \rho=2, \rho_2=2^2$. Let $g(z)\coloneqq f(u)$ where $u=(1-4a)+4a(z+z^{-1})/2$. Since $f$ is analytic on and inside $\tE_{a,\rho_2}$, $g$ is analytic inside the centered disk with radius $\rho_2$. Applying the Hadamard Three Circles Theorem to $g$ gives
	\begin{align*}
		\sup_{z\in \partial\tE_{a,\rho}}\abs{f(z)} \le \tp{\sup_{z\in \partial\tE_{a,\rho_1}}\abs{f(z)}}^{1/2} \tp{\sup_{z\in \partial\tE_{a,\rho_2}}\abs{f(z)}}^{1/2}.
	\end{align*}
	We note that $\tE_{a,\rho_1}$ coincides with the interval $[1-8a,1]$ on the real line. For $z \in \partial\tE_{a,\rho_2}$, \cref{prop:ellipse-bound} implies $\abs{f(z)}\le n\cdot(1+2a(4-1)^2/4)^n \le \exp(5an)$. Therefore
	\begin{align*}
		\sup_{z\in \partial\tE_{a,2}}\abs{f(z)} \le \exp\tp{5an/2}\cdot \tp{\sup_{z\in[1-8a,1]}\abs{f(z)}}^{1/2}.
	\end{align*}
\end{proof}

\subsection{Proof of \cref{lem:technical-main}: A Counting Argument}

We prove \cref{lem:technical-main} in this section. 

We first prove a technical lemma lower bounding the number of binary strings in which all 1s are far away from each other.
\begin{lemma}
    Let $S_{n,r} \subseteq \set{0,1}^{n}$ be the collection of all $n$-bit strings with the property that any two 1's are separated by at least $r$ many 0's. Then $\abs{S_{n,r}}\ge (\sqrt{r+1})^{n/r-1}$.
\end{lemma}
\begin{proof}
For ease of notation we fix $r$ and denote $f(n) \coloneqq \abs{S_{n,r}}$. We observe that $f$ satisfies the following recurrence relation
\begin{align*}
    f(n) = \begin{cases}
        n+1, & \text{ for } 0\le n\le r \\
        f(n-1)+f(n-r-1), & \text{ for } n \ge r+1
    \end{cases}.
\end{align*}
We prove by induction that $f(n)\ge (\sqrt{r+1})^{n/r-1}$. The base case is trivial since $(\sqrt{r+1})^{n/r-1}\le 1$ when $n\le r$.

Now suppose $f(k) \ge (\sqrt{r+1})^{k/r-1}$ for $k\le n-1$. This gives, for $k=n$, the following bound
\begin{align*}
    f(n) &= f(n-1) + f(n-r-1) \\
    &\ge (\sqrt{r+1})^{(n-1)/r-1} + (\sqrt{r+1})^{(n-r-1)/r-1} \\
    &= (\sqrt{r+1})^{(n-1)/r-1}\tp{1+\frac{1}{\sqrt{r+1}}}.
\end{align*}
Since by the AM-GM inequality we have
\begin{align*}
    r+\frac{r}{\sqrt{r+1}} = r-\frac{1}{\sqrt{r+1}}+\sqrt{r+1} \ge \underbrace{1+1+\dots+1}_{r-1\textrm{ 1s}} + \sqrt{r+1} \ge r(\sqrt{r+1})^{1/r},
\end{align*}
or equivalently $1+1/\sqrt{r+1}\ge (\sqrt{r+1})^{1/r}$, we obtain
\begin{align*}
    f(n) \ge (\sqrt{r+1})^{(n-1)/r-1}\cdot (\sqrt{r+1})^{1/r} = (\sqrt{r+1})^{n/r-1}.
\end{align*}
This completes the inductive step, and hence $\abs{S_{n,r}}\ge (\sqrt{r+1})^{n/r-1}$ for all $n, r\in \mathbb{N}$.
\end{proof}

In the following, we fix $k\coloneqq L^{1/3}$, and let $S\coloneqq 0^k \circ S_{L-k,k-1}$. The proof will focus on binary strings in the set $S$. We have $\abs{S}\ge (\sqrt{(k-1)+1})^{(L-k)/(k-1)-1} \ge 2^{(L^{2/3}\log_2 L)/6}$.

Below we characterize some properties of $k$-mer generating polynomials of strings in $S$.
\begin{lemma} \label{lem:set-properties}
Let $S$ be a set of strings defined as above. For $j=1,2,\dots,k$, denote by $\*e_j$ the string with a single ``1'' located at index $j-1$ (indices begin with 0). The following properties hold.
\begin{enumerate}
 \item For any $k$-mer $w \notin \set{0^k, \*e_1, \dots, \*e_k}$, $P_{w,\*x}(z)$ is the zero polynomial.
 \item For any $\*x \in S$ and $1\le j<k$, $P_{\*e_j,\*x}(z) = (p+(1-p)z)\cdot P_{\*e_{j+1},\*x}(z)$.
 \item For any $\*x, \*y \in S$ and $\abs{z}\le 1$, $\abs{P_{0^k, \*x}(z)-P_{0^k, \*y}(z)} \le k\cdot\abs{P_{\*e_k,\*x}(z)-P_{\*e_k,\*y}(z)}$ .
\end{enumerate}
\end{lemma}
\begin{proof}
\textit{Item 1:} By definition of $S$, $\*x[j:j+k-1]$ contains at most one ``1'' for any string $\*x \in S$. Therefore, if $w$ contains at least two ``1''s, then for any $0\le \ell<L-k$,  
\begin{align*}
    K_{w,\*x}[\ell] = \sum_{j=0}^{L-k}\binom{j}{\ell}p^{\ell}(1-p)^{j-\ell}\idc{\*x[j:j+k-1]=w} = 0.
\end{align*}
This means all the coefficients of $P_{w,\*x}(z)$ is zero, and hence $P_{w,\*x}(z)$ is the zero polynomial.

\textit{Item 2:} Since any two consecutive ``1''s in $\*x \in S$ are separated by at least $k-1$ ``0''s, $\*x[i:i+k-1]=\*e_j$ if and only if $\*x[i-1:i+k-2]=\*e_{j+1}$. We thus have
\begin{align*}
    P_{\*e_j,\*x}(z) &= \sum_{i=0}^{L-k} \idc{\*x[i:i+k-1]=\*e_j}\cdot (p+(1-p)z)^i \\
    &= \sum_{i=1}^{L-k} \idc{\*x[i-1:i+k-2]=\*e_{j+1}}\cdot (p+(1-p)z)^i \\
    &= (p+(1-p)z)\cdot \sum_{i=0}^{L-k-1} \idc{\*x[i:i+k-1]=\*e_{j+1}}\cdot (p+(1-p)z)^i \\
    &= (p+(1-p)z)\cdot P_{\*e_{j+1},\*x}(z).
\end{align*}
We have used the fact that for $1\le j< k$, $\idc{\*x[0:k-1]=\*e_j}=\idc{\*x[L-k:L-1]=\*e_{j+1}}=0$.

\textit{Item 3:} We observe that $\*x[i:i+k-1] \in \set{0^k, \*e_1, \dots, \*e_k}$. That implies
\begin{align*}
    \sum_{w \in \set{0^k, \*e_1, \dots, \*e_k}}P_{w,\*x}(z) &= \sum_{w \in \set{0^k, \*e_1, \dots, \*e_k}}\sum_{i=0}^{L-k} \idc{\*x[i:i+k-1]=w}\cdot (p+(1-p)z)^i \\
    &= \sum_{i=0}^{L-k} \tp{\sum_{w \in \set{0^k, \*e_1, \dots, \*e_k}} \idc{\*x[i:i+k-1]=w}}(p+(1-p)z)^i \\
    &= \sum_{i=0}^{L-k}(p+(1-p)z)^i.
\end{align*}
Note the the right-hand-side is independent of $\*x$. Therefore
\begin{align*}
    \abs{P_{0^k,\*x}(z)-P_{0^k,\*y}(z)} &= \abs{\sum_{w \in \set{\*e_1, \dots, \*e_k}}P_{w,\*x}(z) - \sum_{w \in \set{\*e_1, \dots, \*e_k}}P_{w,\*y}(z)} \\
    &\le \sum_{j=1}^{k}\abs{P_{\*e_j,\*x}(z)-P_{\*e_j,\*y}(z)} \\
    &= \sum_{j=1}^{k}\abs{\tp{p+(1-p)z}^{k-j}\cdot\tp{P_{\*e_k,\*x}(z)-P_{\*e_k,\*y}(z)}} \\
    &\le k\cdot \abs{P_{\*e_k,\*x}(z)-P_{\*e_k,\*y}(z)}.
\end{align*}
The second last line is obtained by inductively applying Item 2.
\end{proof}

Below we give the proof of \cref{lem:technical-main}. We use the notations $\exp(x)\coloneqq e^x$, and $\exp_2(x)\coloneqq 2^x$.

\begin{proof}[Proof of \cref{lem:technical-main}]
Let $\*x \in S$ be a string of length $L$. In light of \cref{lem:set-properties}, we only need to consider a fixed $k$-mer $w=\*e_k$, where $k\le L^{1/3}$. Define
\begin{align*}
	g_{\*x}(z) \coloneqq \sum_{j=0}^{L-k}\idc{\*x[j\colon j+k-1]=\*e_k}\cdot z^j.
\end{align*}
Recall that $g_{\*x}(p+qz)=\sum_{j=0}^{L-1}K_{\*e_k,\*x}[j]\cdot z^j=P_{\*e_k,\*x}(z)$. Denote by $a_0(\*x), \dots, a_{L-k}(\*x)$ the Chebyshev coefficients of 
\begin{align*}
	f_{\*x}(z) \coloneqq g_{\*x}(1-4a+4a\cdot z),
\end{align*}
where $a\coloneqq L^{-2/3}\log^{1/4}L$ (equivalently, the coordinates of $f_{\*x}$ in the Chebyshev basis). In other words, we can write
\begin{align*}
	f_{\*x}(z) = \sum_{j=0}^{L-k}a_j(\*x) \cdot T_j(z).
\end{align*}
We first argue that only the first few coefficients are significant. This can be done by applying \cref{thm:chebcoef-bound} to $f_{\*x}(z)$, say, with $\rho=2$. To that end, we first upper bound $\abs{f_{\*x}(z)}$ for $z \in E_{2}$. Denoting $z'=1-4a+4a\cdot z$, we have that $z' \in \tE_{a,2}$ when $z\in E_2$. By item (1) of \cref{prop:ellipse-bound}, we have $\abs{z'}\le 1+a$. It follows that
\begin{align*}
    \sup_{z \in E_2}\abs{f_{\*x}(z)} = \sup_{z' \in \tE_{a,2}}\abs{g_{\*x}(z')} \le L(1+a)^L \le L\exp(aL) = L\exp(L^{1/3}\log^{1/4}L).
\end{align*}
Therefore, we can apply \cref{thm:chebcoef-bound} to $f_{\*x}(z)$ with $\rho=2$,$M=L\exp(L^{1/3}\log^{1/4}L)$ and get (for large enough $L$)
\begin{align*}
	\forall j\ge L^{1/3}\sqrt{\log L},\quad \abs{a_{j}(\*x)} \le L\exp(L^{1/3}\log^{1/4}L)\cdot 2^{-L^{1/3}\sqrt{\log L}} \le 2^{-L^{1/3}\sqrt{\log L}/8}.
\end{align*}

To each string $\*x \in \set{0,1}^{L}$ we associate a vector
\begin{align*}
	\phi(\*x) \coloneqq \tp{a_j(\*x)\colon j=0,1,\dots,L^{1/3}\sqrt{\log L}-1}.
\end{align*}
\cref{prop:trivial-chebcoef-bound} implies each entry of $\phi(\*x)$ belongs to the interval $[-2(L-k+1),2(L-k+1)]\subseteq [-2L,2L]$. We now partition $[-2L, 2L]$ into $m$ smaller intervals $I_1, \dots, I_m$, each of length $2^{-L^{1/3}\sqrt{\log L}/8}$, meaning that $m=4L\cdot 2^{L^{1/3}\sqrt{\log L}/8}$. The vector $\phi(\*x)$ must fall into one of the sub-cubes of the form
\begin{align*}
	\+I(r) \coloneqq \prod_{0\le j< L^{1/3}\sqrt{\log L}}I_{r(j)},
\end{align*}
where $r\colon [L^{1/3}\sqrt{\log L}] \rightarrow [m]$ is a mapping that uniquely identifies the sub-cube. It follows that the total number of such sub-cubes is
\begin{align*}
	m^{L^{1/3}\sqrt{\log L}} &= \tp{4L\cdot 2^{L^{1/3}\sqrt{\log L}/8}}^{L^{1/3}\sqrt{\log L}} \\
	&= \exp_2\tp{\tp{\frac{L^{1/3}\sqrt{\log L}}{8} + \log_2{L} + 2}\cdot L^{1/3}\sqrt{\log L}} \\
	&\le \exp_2\tp{\frac{L^{2/3}\log L}{8} + O(L^{1/3}\log^{3/2}{L})} < 2^{(L^{2/3}\log L)/6} \le \abs{S}
\end{align*}
for large enough $L$. By the Pigeonhole Principle, there must be two distinct strings $\*x, \*y \in S$ such that $\phi(\*x), \phi(\*y)$ fall into the same sub-cube. In other words, we have
\begin{align*}
	\forall 0\le j< L^{1/3}\sqrt{\log L}, \quad \abs{a_j(\*x)-a_j(\*y)} \le 2^{-L^{1/3}\sqrt{\log L}/8}.
\end{align*} 
It follows that
\begin{align*}
	\sup_{z \in [1-8a,1]}\abs{g_{\*x}(z)-g_{\*y}(z)} &= \sup_{z \in [-1,1]}\abs{f_{\*x}(z)-f_{\*y}(z)} \\ 
	&\le \sup_{z \in [-1,1]}\sum_{j=0}^{L-k}\abs{a_j(\*x)-a_j(\*y)}\cdot\abs{T_j(z)} \\
	&\le \sum_{j=0}^{L^{1/3}\sqrt{\log L} - 1}2^{-L^{1/3}\sqrt{\log L}/8} + \sum_{j= L^{1/3}\sqrt{\log L}}^{L-k}2\cdot 2^{-L^{1/3}\sqrt{\log L}/8} \\
	&\le 2^{-L^{1/3}\sqrt{\log L}/7}.
\end{align*}
Applying \cref{cor:ellipse-to-interval} to $g_{\*x}-g_{\*y}$ with $a=L^{-2/3}\log^{1/4}L$ gives (for large enough $L$)
\begin{align*}
	\sup_{z \in \partial \tE_{a,2}}\abs{g_{\*x}(z)-g_{\*y}(z)} &\le \exp\tp{5aL/2} \cdot \sup_{z \in [1-8a,1]}\abs{g_{\*x}(z)-g_{\*y}(z)} \\
	&\le \exp\tp{5L^{1/3}\log^{1/4}L/2} \cdot 2^{-L^{1/3}\sqrt{\log L}/14} \le 2^{-L^{1/3}\sqrt{\log L}/15}.
\end{align*}
Let $\Gamma$ be the sub-arc of the circle $\set{p+qz\colon \abs{z}=1}$ which lies completely inside the ellipse $\tE_{a,2}$. Item (2) of \cref{prop:ellipse-bound} implies that the length of $\Gamma$ is at least $a=L^{-2/3}\log^{1/4}L$. Therefore the Maximum Modulus Principle implies
\begin{align*}
	\sup_{\theta\colon \abs{\theta}\le a}\abs{P_{\*e_k,\*x}(e^{i\theta})-P_{\*e_k,\*y}(e^{i\theta})} \le \sup_{z\in\Gamma}\abs{g_{\*x}(z)-g_{\*y}(z)} \le \sup_{z\in\tE_{a,2}}\abs{g_{\*x}(z)-g_{\*y}(z)} \le 2^{-L^{1/3}\sqrt{\log L}/15}.
\end{align*}
Now we have established the lemma for a fixed $k$-mer $w=\*e_k$. Since $\*x, \*y \in S$, \cref{lem:set-properties} says that for any other $k$-mer $w \in \set{0,1}^k$ either both $P_{w,\*x}(z)$ and $P_{w,\*y}(z)$ are zero polynomials or $w \in \set{0^k, \*e_1, \dots, \*e_k}$ and
\begin{align*}
    \sup_{\theta\colon \abs{\theta}\le a}\abs{P_{w,\*x}(e^{i\theta})-P_{w,\*y}(e^{i\theta})} \le k\cdot \sup_{\theta\colon \abs{\theta}\le a}\abs{P_{\*e_k,\*x}(e^{i\theta})-P_{\*e_k,\*y}(e^{i\theta})} \le 2^{-L^{1/3}\sqrt{\log L}/20}.
\end{align*}
Finally, we note that both $\*x$ and $\*y$ start with a run of 0's of length $k=L^{1/3}$.

\end{proof}

\begin{remark}
A much simpler proof for the slightly weaker bound $2^{-\Omega(L^{1/3})}$ is possible based on the complex analytical result of Borwein and Erdelyi \cite[Theorem 3.3]{borwein1997littlewood} (see also \cite{de2019optimal, nazarov2017trace}): there exist strings $\*x, \*y \in \set{0,1}^{L^{2/3}}$ such that
\begin{align*}
    2^{-\Omega(L^{1/3})} \ge \sup_{z \in \Gamma_{L^{-1/3}}} \abs{P_{\*x}(z)-P_{\*y}(z)} = \sup_{z \in \Gamma_{L^{-2/3}}} \abs{P_{\*x}(z^{L^{1/3}})-P_{\*y}(z^{L^{1/3}})},
\end{align*}
where $P_{\*x}(z)\coloneqq\sum_{j=0}^{\abs{\*x}-1}x_jz^j$, $\Gamma_a$ stands for the sub-arc $\set{e^{i\theta}\colon \abs{\theta} < a}$. Now we observe that $P_{\*x}(z^{L^{1/3}})=P_{\*x'}(z)$ where $\*x' \in \set{0,1}^{L}$ is the string obtained by inserting $L^{1/3}-1$ many 0's before every bit of $\*x$ ($\*y'$ is defined similarly). Clearly, $\*x', \*y' \in S$ since any two 1's are separated by at least $L^{1/3}-1$ many 0's. Therefore, they enjoy the properties in \cref{lem:set-properties}, and \cref{lem:complex-main} follows with a weaker bound.\footnote{We thank an anonymous reviewer for pointing this observation out to us.}
\end{remark}



\section{Optimality of the Maximum Likelihood Estimation}\label{sec:mle}

\begin{proof}[Proof of \thmref{thm:mle-optimality}]
	For $1\le i\le m$ define $S_i\coloneqq \set{x \in \Omega \colon D_0(x) > D_i(x)}$, and let $S\coloneqq S_1\cap S_2 \cap \dots \cap S_m$. By definition of the total variation distance, we have 
	\begin{align*}
		1-\eps \le \tv{D_0, D_i} = D_0(S_i) - D_i(S_i) \le D_0(S_i).
	\end{align*}
	The Union Bound thus implies $D_0(S)\ge 1-m\eps$. Moreover, by \cref{def:mle}, when $x \in S$ it must hold that $\mle(x;\+D)=0$. Therefore
	\begin{align*}
		\PRs{x\sim\+D_0}{\mle(x;\+D)=0} \ge \PRs{x\sim\+D_0}{x \in S} = D_0(S) \ge 1-m\eps.
	\end{align*}
\end{proof}

\begin{proof}[Proof of Corollary \ref{cor:mle-trace}]
The Chernoff bound implies that if we repeat the purported reconstruction algorithm $8\ln(1/\eps)n$ times and output the majority, it succeeds with probability at least $1-\eps/2^{n+1}$.

Let $A$ be such a (deterministic) reconstruction algorithm with $T'=8\ln(1/\eps)\cdot nT$ traces described as above, which successfully outputs the source string $x$ with probability at least $1-\eps/2^{n+1}$. Formally, for any source string $x\in\set{0,1}^n$, it holds that
\begin{align*}
    \PRs{\tx_1, \dots, \tx_{T'} \sim D_x}{A(\tx_1, \dots, \tx_{T'}) = x} \ge 1-\eps/2^{n+1}.
\end{align*}
Let $R_x \subseteq \tp{\set{0,1}^{\le n}}^{T'}$ be exactly the collection of $T'$-tuples of strings on which $A$ outputs $x$. We thus have
\begin{align*}
    \forall x\in\set{0,1}^n,\quad D_{x}^{\otimes T'}(R_x) \ge 1-\eps/2^{n+1},
\end{align*}
where $D_x^{\otimes T'}$ denotes the $T'$-fold product of $D_x$ with itself, capturing the distribution of $\tp{\tx_1,\dots,\tx_{T'}}$. On the other hand, for distinct strings $x$ and $y$ we have $R_x \cap R_y = \varnothing$ (by definition, $A$ cannot output both $x$ and $y$ on the same input), and hence the bound
\begin{align*}
    D_y^{\otimes T'}(R_x) \le 1 - D_y^{\otimes T'}(R_y) \le \eps.
\end{align*}
This implies
\begin{align*}
    \tv{D_x^{\otimes T'}, D_y^{\otimes T'}} = \sup_{S}\abs{D_{x}^{\otimes T'}(S) - D_y^{\otimes T'}(S)} \ge D_{x}^{\otimes T'}(R_x) - D_y^{\otimes T'}(R_x) \ge 1-2\eps/2^{n+1}=1-\eps/2^n.
\end{align*}
We stress that the above bound holds for any pair of distinct strings $x, y \in \set{0,1}^n$. Applying \cref{thm:mle-optimality} to $\+D\coloneqq\set{D_x^{\otimes T'}\colon x\in\set{0,1}^n}$ gives
\begin{align*}
    \forall x \in \set{0,1}^n,\quad \PRs{\tx_1,\dots,\tx_{T'}\sim D_x}{\mle(\tx_1,\dots,\tx_{T'}; \+D)=x} \ge 1 - (2^n-1)\cdot \eps/2^n \ge 1 - \eps.
\end{align*}
\end{proof}

\begin{proof}[Proof of \thmref{thm:mle-lb}]
 The distributions are defined as follows. Let $t=\floor{n/4}$, and so $m=\binom{n}{t}$. The domain $\Omega = \Omega_1\cup \Omega_2$ where $\Omega_1=\binom{[n]}{t}$ is the collection of all subsets of $[n]$ of size exactly $t$, and $\Omega_2=[n]$. We have 
 \begin{align*}
    \abs{\Omega} = \binom{n}{t} + n = m + n.
 \end{align*}
 
 We first define $D_0$ to be the uniform distribution over $\Omega_2$, \ie, $D_0(\mathfrak{s})=1/n$ for any $\mathfrak{s} \in [n]$.
 
 For each one of the remaining $m$ distributions, we identify it with a $t$-subset $S \in \binom{[n]}{t}$. The precise definition of $D_S$ is as follows.
 \begin{align*}
    \forall S \in \binom{[n]}{t}, \quad D_S(\mathfrak{s}) = \begin{cases}
                                                             2/3 & \textup{if $\mathfrak{s} \in \Omega_1$, and $\mathfrak{s} = S$} \\
                                                             1/(3t) & \textup{if $\mathfrak{s} \in \Omega_2$, and $\mathfrak{s} \in S$} \\
                                                             0 & \textup{otherwise}
                                                            \end{cases}.
 \end{align*}
 In other words, $\mathfrak{s} \in \Omega_1$ occurs with probability $2/3$, conditioned on which $D_S$ is the point distribution supported on $\set{S}$; $\mathfrak{s} \in \Omega_2$ occurs with probability $1/3$, conditioned on which $D_S$ is the uniform distribution over $S$. Now we verify that $\+D=\set{D_0,D_1,\dots, D_m}$ satisfies the two conditions.
 
 For Condition 1, consider a distinguisher $A$ which on sample $\mathfrak{s} \in \Omega$, outputs $S$ if $\mathfrak{s} = S \in \Omega_1$, and outputs 0 if $\mathfrak{s} \in \Omega_2$. We have
 \begin{align*}
    \PRs{\mathfrak{s}\sim D_0}{A(\mathfrak{s})=0}=1\ge 2/3, \quad \PRs{\mathfrak{s}\sim D_S}{A(\mathfrak{s})=S}=2/3.
 \end{align*}
 
 To see Condition 2, let $\mathfrak{s}_1, \dots, \mathfrak{s}_T \sim D_0$ be $T\le \floor{n/4}$ samples. Since $D_0$ is supported on $\Omega_2=[n]$, the samples are all elements of $[n]$, meaning that there is at least one $S \in \binom{[n]}{t}$ containing all samples. Calculating the likelihoods gives
 \begin{align*}
    \prod_{i=1}^{T}D_S(\mathfrak{s}_i) = \tp{\frac{1}{3t}}^T \ge \tp{\frac{4}{3n}}^T > \tp{\frac{1}{n}}^T = \prod_{i=1}^{T}D_0(\mathfrak{s}_i).
 \end{align*}
 Therefore, the output of the Maximum Likelihood Estimation on $\mathfrak{s}_1, \dots, \mathfrak{s}_T \sim D_0$ will never be $0$.

\end{proof}

\begin{section}{Acknowledgements}
    We are thankful to several anonymous reviewers for their valuable suggestions and comments.
\end{section}

\bibliographystyle{alpha}
\bibliography{reference}

\end{document}